\documentclass[lettersize,journal]{IEEEtran}
\usepackage{amsmath,amsfonts}
\usepackage{algorithmic}
\usepackage{algorithm}
\usepackage{array}
\usepackage[caption=false,font=normalsize,labelfont=sf,textfont=sf]{subfig}
\usepackage{textcomp}
\usepackage{stfloats}
\usepackage{url}
\usepackage{verbatim}
\usepackage{graphicx}
\usepackage{cite}

\usepackage{microtype}
\usepackage{tumcolor}
\usepackage{amssymb,amsthm,fixmath}
\usepackage{mathtools}
\usepackage{xcolor}
\usepackage{bm}
\usepackage{pgfplots}
\usepackage{tikz}
\usetikzlibrary{intersections,pgfplots.fillbetween,patterns,spy}
\usetikzlibrary{calc,fit,calc,plotmarks,external,shapes,arrows,decorations.markings}
\usetikzlibrary{decorations.pathmorphing,calc,shapes,shapes.geometric}
\usepgfplotslibrary{groupplots,dateplot}
\usetikzlibrary{patterns,shapes.arrows}
\usetikzlibrary{backgrounds,positioning,fit,matrix,fillbetween}
\usepackage{cuted}
\usepackage{tabularx}
\usepackage{soul}
\usepackage{hyperref}
\usepackage[acronym,shortcuts]{glossaries}
\hypersetup{
	colorlinks=true,       % false: boxed links; true: colored links
	linkcolor=black,          % color of internal links (change box color with linkbordercolor)
	citecolor=black,        % color of links to bibliography
	filecolor=black,         % color of file links
	urlcolor=black        % color of external links
}
\usepackage{cleveref}

\definecolor{mylila}{RGB}{153,50,204}

\newcommand{\op}[1]{{\operatorname{#1}}}
\newcommand{\B}[1]{{\bm{#1}}}
\newcommand{\uproman}[1]{\uppercase\expandafter{\romannumeral#1}}

\newcommand{\h}{^{\operatorname{H}}}
\newcommand{\T}{^{\operatorname{T}}}
\newcommand{\inv}{^{-1}}
\newcommand{\diag}{\operatorname{diag}}
\DeclareMathOperator{\eye}{\mathbf{I}}
\newcommand{\vect}{\operatorname{vec}}
\newcommand{\tr}{\operatorname{tr}}

\newcommand{\C}{\mathbb{C}}

\DeclareMathOperator*{\argmax}{arg\,max}

\newacronym{mimo}{MIMO}{multiple-input multiple-output}
\newacronym{simo}{SIMO}{single-input multiple-output}
\newacronym{siso}{SISO}{single-input single-output}
\newacronym{mse}{MSE}{mean square error}
\newacronym{cme}{CME}{conditional mean estimator}
\newacronym{pdf}{PDF}{probability density function}
\newacronym{adc}{ADC}{analog-to-digital converter}
\newacronym{mmse}{MMSE}{minimum mean square error}
\newacronym{snr}{SNR}{signal-to-noise ratio}
\newacronym{evd}{EVD}{eigenvalue decomposition}
\newacronym{crb}{CRB}{Cram\'er-Rao bound}
\newacronym{map}{MAP}{maximum a posteriori}
\newacronym{cs}{CS}{compressive sensing}
\newacronym{ls}{LS}{least squares}
\newacronym{awgn}{AWGN}{additive white Gaussian noise}
\newacronym{csi}{CSI}{channel state information}
\newacronym{pmf}{PMF}{probability mass function}
\newacronym{cdf}{CDF}{cumulative distribution function}
\newacronym{rv}{RV}{random variable}
\newacronym{gmm}{GMM}{Gaussian mixture model}
\newacronym{bs}{BS}{base station}
\newacronym{ofdm}{OFDM}{orthogonal frequency-division multiplexing}
\newacronym{ula}{ULA}{uniform linear array}
\newacronym{uma}{UMa}{urban macrocell}
\newacronym{em}{EM}{expectation-maximization}
\newacronym{lte}{LTE}{long term evolution}
\newacronym{vae}{VAE}{variational auto-encoder}
\newacronym{omp}{OMP}{orthogonal matching pursuit}
\newacronym{dft}{DFT}{discrete Fourier transform}
\newacronym{ml}{ML}{machine learning}
\newacronym{cnn}{CNN}{convolutional neural network}
\newacronym{mt}{MT}{mobile terminal}

\newcommand{\quadriga}{QuaDRiGa}

\newcolumntype{Y}{>{\centering\arraybackslash}X}

\colorlet{TUMOrange80}{TUMOrange!80}%
\colorlet{TUMOrange60}{TUMOrange!60}%
\colorlet{TUMOrange50}{TUMOrange!50}%
\colorlet{TUMOrange40}{TUMOrange!40}%
\colorlet{TUMOrange20}{TUMOrange!20}%

\newcommand{\marksize}{1.6pt}
\newcommand{\lineWidth}{1.2pt}

\tikzset{genie/.style={mark options={solid},color=black, line width=\lineWidth, mark size=\marksize}}
\tikzset{LS/.style={mark options={solid},color=black, line width=\lineWidth, mark size=\marksize,dashed}}
\tikzset{vae/.style={mark options={solid},color=TUMBeamerOrange, line width=\lineWidth, mark=x, mark size=\marksize,dashed}}

\tikzset{cnn1/.style={mark options={solid},color=TUMBeamerGreen, line width=\lineWidth, mark=diamond, mark size=\marksize,dashdotted}}

\tikzset{gmm/.style={mark options={solid},color=black, line width=\lineWidth, mark=square, mark size=\marksize}}
\tikzset{gmm_perSNR/.style={mark options={solid},color=TUMBlue, line width=\lineWidth, mark=triangle, mark size=\marksize,dotted}}
\tikzset{gmm_20dB/.style={mark options={solid},color=TUMBeamerRed, line width=\lineWidth, mark=o, mark size=\marksize,dashed}}
\tikzset{gmm_toep/.style={mark options={solid},color=TUMBeamerRed, line width=\lineWidth, mark=square, mark size=\marksize,dotted}}

\tikzset{samplecov/.style={mark options={solid},color=TUMOrange, line width=\lineWidth, mark=diamond, mark size=\marksize,dashdotted}}

\tikzset{
	treetop/.style = {decoration={random steps, segment length=0.4mm}, decorate},
	trunk/.style   = {decoration={random steps, segment length=2mm,
			amplitude=0.2mm}, decorate}}

\tikzset{
	my tree/.pic={
		\foreach \w/\f in {0.3/30,0.2/50,0.1/70} {
			\fill [brown!\f!black, trunk] (-\w/8,0) rectangle +(\w/4,1);
		}
		\foreach \n/\f in {0.4/40,0.3/60} {
			\fill [green!\f!black, treetop](0,1) ellipse (\n/1.5 and \n);
		}
	}
}

\tikzset{BS style/.style={
		draw,
		minimum width=0.5cm,
		minimum height=2cm,
		fill=gray,
		inner sep=0,
		outer sep=0,
}}
\tikzset{UE Style/.style={
		draw,
		fill=gray,
		minimum width=0.2cm,
		minimum height=0.3cm}
}

\newcommand{\plotwidth}{1\columnwidth}
\newcommand{\plotheight}{0.55\columnwidth}
\newacronym{psd}{PSD}{positive semi-definite}
\usepackage{soul}

\newtheorem{theorem}{Theorem}

\newcommand{\blue}[1]{\textcolor{black}{#1}}

\newacronym{relu}{ReLU}{rectified linear unit}

\begin{document}
	
	%\title{Unsupervised Generative Modeling for Robust Channel Estimation with Gaussian Mixture Models}
	%\title{Learning a Generative Model for Robust Channel Estimation from Imperfect Data}
	\title{Learning a Gaussian Mixture Model from Imperfect Training Data for Robust Channel Estimation}
	\author{Benedikt Fesl, Nurettin Turan,~\IEEEmembership{Student Member,~IEEE,} Michael Joham,~\IEEEmembership{Member,~IEEE,} \\and Wolfgang Utschick,~\IEEEmembership{Fellow,~IEEE}
		\vspace{-0.5cm}
		% <-this % stops a space
		\thanks{This work was partly funded by Huawei Sweden Technologies AB, Lund.}% <-this % stops a space
		%\thanks{Manuscript received April 19, 2021; revised August 16, 2021.}
		\thanks{
			The authors are with Professur f\"ur Methoden der Signalverarbeitung, Technische Universit\"at M\"unchen, 80333 M\"unchen, Germany  (e-mail: benedikt.fesl@tum.de; nurettin.turan@tum.de; joham@tum.de; utschick@tum.de).
		}
	}
	
	% The paper headers
	%\markboth{Journal of \LaTeX\ Class Files,~Vol.~14, No.~8, August~2021}%
	%{Shell \MakeLowercase{\textit{et al.}}: A Sample Article Using IEEEtran.cls for IEEE Journals}
	
	%\IEEEpubid{0000--0000/00\$00.00~\copyright~2021 IEEE}
	% Remember, if you use this you must call \IEEEpubidadjcol in the second
	% column for its text to clear the IEEEpubid mark.
	
	\maketitle

	\begin{abstract}
		In this letter, we propose a Gaussian mixture model (GMM)\glsunset{gmm}-based channel estimator which is learned on imperfect training data, i.e., the training data \blue{are} solely comprised of noisy and sparsely allocated pilot observations. 
		In a practical application, recent pilot observations at the \ac{bs} can be utilized for training.
		%		The proposed estimator allows for a practical application where recent pilot observations at the \ac{bs} can be utilized for training.
		This is in sharp contrast to state-of-the-art \ac{ml} techniques where a \blue{training} dataset consisting of perfect \ac{csi} \blue{samples} is a prerequisite, which is generally unaffordable.
		In particular, we propose an adapted training procedure for fitting the \ac{gmm} which is a generative model that represents the distribution of all potential channels associated with a specific \ac{bs} cell.
		%		Afterwards, the trained \ac{gmm} is leveraged for channel estimation.
		%		To account for the imperfections in the training data, the underlying \ac{em} algorithm to fit the \ac{gmm} parameters is adapted to the system model. 
		To this end, the necessary modifications of the underlying \ac{em} algorithm
		%		for both a spatial and a doubly-selective fading model 
		are derived.
		%		To account for the imperfections in the training data, an adaptation of the \ac{em} algorithm is utilized in order to train with pilot observations which are corrupted by \ac{awgn}, or which stem from a \ac{ofdm} grid with a fixed pilot pattern.
		Numerical results show that the proposed estimator performs close to the case where perfect \ac{csi} is available for the training and exhibits a higher robustness against imperfections in the training data as compared to state-of-the-art \ac{ml} techniques.
		%		. Further, the proposed method exhibits higher robustness against imperfections in the training data as compared to state-of-the-art \ac{ml} techniques.
		%		which learn a direct mapping from pilots to corresponding channel estimates. 
	\end{abstract}
	

	\section{Introduction}
			\begin{figure}[b]
		\onecolumn
		\centering
		\copyright \scriptsize{This work has been submitted to the IEEE for possible publication. Copyright may be transferred without notice, after which this version may no longer be accessible.}
		\vspace{-1.3cm}
		\twocolumn
	\end{figure}
	\IEEEPARstart{C}{hannel} estimation plays a crucial role in enhancing wireless communications systems. Recently, \ac{ml} approaches were successfully leveraged for channel estimation.
	The goal is to exploit a priori information about all possible channels of \acp{mt} associated with a specific \ac{bs} cell and its radio propagation environment to improve the channel estimation quality.
	%	\ac{bs} environment, i.e., the radio propagation scenario as shown in Fig. \ref{fig:radio_env}, 
	This is generally intractable to model analytically but is represented in terms of a training dataset that is available at the \ac{bs}.
	%	, e.g., \cite{NeWiUt18,8052521,8640815,vae_chEst}. 
	
	Thereby, different learning techniques can be distinguished. In end-to-end learning, the channel estimation is not performed explicitly, but a network is trained to directly perform signal detection \cite{8052521}. 
	A different approach is to learn a nonlinear regression mapping from the pilot observation in the input to a channel estimate at the output of a network by utilizing ground-truth \ac{csi} \blue{samples} \cite{8640815,NeWiUt18}. In contrast to that, it was recently proposed to train a generative model that represents the channel distribution of the whole \ac{bs} cell, which is afterwards leveraged for channel estimation \cite{9842343,vae_chEst}. 
	%	This was shown to perform well using either a \ac{gmm} \cite{9842343}, or a \ac{vae} \cite{vae_chEst}.
	
	A common prerequisite of \ac{ml}-based approaches is the availability of a representative training dataset consisting of perfect \ac{csi} \blue{samples}.
	However, the construction of such a dataset is a challenging task in practice. One possibility is to perform costly measurement campaigns for each \ac{bs}, which is generally unaffordable. A different attempt is to use channel simulators, e.g., \cite{QuaDRiGa1}. %e.g., stochastic-geometric models \cite{QuaDRiGa1} or ray-tracing tools. 
	The inherent problem of simulators is the mismatch between the real and the \blue{artificially} generated data, \blue{leading to performance losses}. 
	Another problem is that the \ac{bs} environment may change over time which is difficult to track.
	%because of, e.g., different weather conditions
	This leads to the idea of utilizing pilot observations, which are received in great numbers at the \ac{bs} during regular operation and capture the environmental information, as training data.
	Although the data aggregation is then cheap and the dataset can be continuously updated to account for varying conditions, the dataset contains imperfections, e.g., noise and sparse pilot allocations. However, these imperfections can be mitigated by specific training adaptations since they follow a known model. 
	%	Since no ground-truth \ac{csi} labels are available in this setup, it is an unsupervised task.
	
	%	\begin{figure}[t]
		%		\centering
		%		\input{graphics/radio_env}
		%		\caption{Visualization of an urban radio propagation environment of a \ac{bs} cell. The different links of the \acp{mt} are highlighted in red.}
		%		\label{fig:radio_env}
		%	\end{figure}
	
	%	A related prior work for unsupervised channel estimation with generative models is \cite{vae_chEst}, where a \ac{vae} is fitted purely on data containing \ac{awgn}. However, this approach is not applied for \ac{ofdm} systems where pilot data stem from pre-defined allocation patterns.
	
	%	\subsection{Contributions}
	\textit{Contributions:}
	In this letter, we propose an adaptation of the \ac{em} algorithm to fit a \ac{gmm} with noisy, and possibly sparsely allocated, pilot observations. 
	%	The \ac{gmm} is afterwards leveraged for channel estimation. 
	We derive new update steps for the \ac{gmm} parameters that take the system model, i.e., the model of the imperfections, into account. Additionally, we show that imposing structural features to the covariances of the \ac{gmm} acts as a regularization which enhances the channel estimation performance, especially for sparse pilot allocations.
	It is discussed that \blue{the presented method} is particularly robust against imperfections in the training data, in contrast to commonly used regression-based \ac{ml} techniques. Finally, the discussed properties are verified with simulations for both a spatial and a doubly-selective fading model.
	%	where a channel estimation quality close to the case of perfect training \ac{csi} and a performance gain over state-of-the-art approaches can be attested.

	%	\begin{itemize}
		%		\item GMM estimator \cite{9842343,Koller2021icassp,gmm_structured} 
		%		\item OMP unfolding NN channel estimator \cite{9690064}
		%		\item VAE estimator \cite{vae_chEst}
		%		\item ChannelNet \cite{8640815}, ChannelNet + CAE \cite{9166541}
		%		\item EM with missing data \cite[Chapter 11]{bookRubin02}
		%		\item EM for structured covariances \cite{342500}
		%		\item Deep-learning for doubly-selective fading: \cite{8672767}
		%		\item Learning the MMSE Channel estimator \cite{NeWiUt18}
		%		\item Can it be shown that the likelihood is continuously maximized with the projection method?
		%	\end{itemize}
	\blue{
		\textit{Notation:}
		The identity matrix of dimension $N\times N$ is denoted by $\eye_N$.
		The column-wise vectorization and the trace of a matrix is denoted by $\vect(\cdot)$ and $\tr(\cdot)$, respectively. 
		%	The trace of a matrix is written as $\tr(\cdot)$.
		%	 We denote the circularly symmetric complex Gaussian distribution with mean $\B \mu$ and covariance $\B C$, evaluated at $\B x$, as $\mathcal{N}_{\mathbb{C}}(\B x ; \B \mu, \B C)$. 
		%	The operation $\diag(\cdot)$ denotes both the diagonal vector of a matrix and a diagonal matrix with the vector on its diagonal.
		%	A diagonal matrix with diagonal $\B x$ is denoted by $\diag(\B x)$, whereas the vector containing the diagonal elements of a matrix $\B X$ is given as $\diag(\B X)$. 
		We denote a \ac{psd} matrix $\B C$ which fulfills $\B x\h \B C \B x \geq 0$ for all $\B x\neq \B 0$ as $\B C \succeq \B 0$. 
	}
	\section{System and Channel Models}
	
	We consider the \blue{general} system model of pilot observations
	\begin{equation}
		\B y = \B A \B h + \B n ,
		\label{eq:system_model}
	\end{equation}
	where $\B A\blue{\in\mathbb{C}^{M\times N}}$ is the known observation matrix, $\B h\in\C^N$ is the wireless channel with an unknown probability density function (PDF)\glsunset{pdf} $f_{\B h}$, and \ac{awgn} $\B n\sim\mathcal{N}_\C(\B 0, \B C_{\B n}=\sigma^2\eye_{\blue{M}})$. 
	%	This general model \eqref{eq:system_model} can describe both a spatial and a \ac{ofdm} system model which are introduced in the following.
	In this letter, we consider the following two instances of \eqref{eq:system_model}.
	
	\subsection{Spatial System Model}\label{subsec:spatial_channel_model}
	Consider a \ac{simo} system where a \ac{bs} equipped with $N$ antennas serves single-antenna \acp{mt} in a \blue{single-user} uplink transmission. \blue{After correlating the pilot sequence for a single snapshot we obtain $\B y =\B h + \B n$, following the model in \eqref{eq:system_model} with $M=N$ and $\B A = \eye_N$.}
	
	We work with a spatial channel model \cite{NeWiUt18} where channels are modeled conditionally Gaussian: $ \B h | \B \delta \sim \mathcal{N}_\C(\B 0, \B C_{\B \delta}) $.
	The random vector $ \B \delta $ collects the angles of arrival and path gains of the main propagation clusters between a \ac{mt} and the \ac{bs}. 
	A different spatial channel covariance matrix $ \B C_{\B \delta}$ is computed for each sample by means of the steering vector of a \ac{ula} antenna array at the \ac{bs}, cf. \cite{NeWiUt18}.
	%	The main angles are drawn independently and uniformly from the interval $ [0, 2\pi] $ and the path gains are independent zero-mean Gaussians.
	%	The \ac{bs} employs a \ac{ula} such that the spatial channel covariance matrix is given by
	%	\begin{equation}
		%		\B C_{\B \delta}= \int_{-\pi}^\pi g(\theta; \B \delta) \B a(\theta) \B a(\theta)\h d \theta,
		%	\end{equation}
	%	where $	\B a(\theta) = [1, e^{j\pi\sin(\theta)}, \dots, e^{j\pi(N-1)\sin(\theta)}]\T$
	%	is the array steering vector for an angle of arrival $ \theta $ and $ g $ is a power density function~\cite{3gpp}.
	
	%	consisting of a sum of weighted Laplace densities whose standard deviations describe the angle spread of the propagation clusters~\cite{3gpp}.
	%	For every channel sample, we generate random angles and path gains, combined in $\B \delta$, and then draw the sample as $ \B h \sim \mathcal{N}_\C(\B 0, \B C_{\B \delta}) $.

	\subsection{Doubly-Selective Fading System Model}\label{subsec:doubly_channel_model}
	
	%	\begin{figure}[t]
		%		\begin{center}
			%			\newcommand*{\xMin}{1}%
			%			\newcommand*{\xMax}{15}%
			%			\newcommand*{\yMin}{1}%
			%			\newcommand*{\yMax}{13}%
			%			\begin{tikzpicture}[scale=0.22]
				%				\foreach \i in {\xMin,...,\xMax} {
					%					\draw [very thin,gray] (\i,\yMin) -- (\i,\yMax) ;
					%				}
				%				\foreach \i in {\yMin,...,\yMax} {
					%					\draw [very thin,gray] (\xMin,\i) -- (\xMax,\i);
					%				}
				%				\foreach \i in {1,4,8,12} {
					%					\node at (0.3, \i+0.5) {\scriptsize$\i$};
					%				}
				%				\foreach \i in {1,5,10,14} {
					%					\node at (\i+0.5,0.3) {\scriptsize$\i$};
					%				}
				%				\foreach \i in {0, 6, 13} {
					%					\foreach \k in {0, 4, 8, 11} {
						%						\draw[fill=TUMBeamerBlue] (\i+1,\k+1) -- (\i+2,\k+1) -- (\i+2,\k+2) -- (\i+1,\k+2) --(\i+1,\k+1);
						%					}
					%				}
				%				\foreach \i in {3, 9} {
					%					\foreach \k in {2, 6, 9} {
						%						\draw[fill=TUMBeamerBlue] (\i+1,\k+1) -- (\i+2,\k+1) -- (\i+2,\k+2) -- (\i+1,\k+2) --(\i+1,\k+1);
						%					}
					%				}
				%				\node at (15/2, -0.8) {\small Time symbol};
				%				\node[rotate=90] at (-0.8,13/2) {\small Carrier};
				%				\draw[fill=TUMBeamerBlue] (16,12)--(17,12)--(17,13)--(16,13)--(16,12) node [right] at (17,12.5) {\small Pilot};
				%				\draw[very thin,gray] (16,10)--(17,10)--(17,11)--(16,11)--(16,10) node [right,color=black] at (17,10.5) {\small Data};
				%			\end{tikzpicture}
			%		\vspace{-0.3cm}
			%		\end{center}
		%		\caption{Diamond-shaped pilot allocation scheme for $N_p=18$ pilots.}
		%		%\vspace{-0.2cm}
		%		\label{fig:pilot_alloc}
		%	\end{figure}
	
	In this case, we consider a \ac{siso} transmission in the spatial domain over a doubly-selective fading channel \blue{$\B h = \vect(\B H)$}, where $\B H\in \C^{N_c\times N_t} $ represents the time-frequency response of the channel for $ N_c $ carriers and $ N_t $ time slots. This is a typical setup in \ac{ofdm} systems.
	%We consider a \ac{dl} \ac{fdd} system where the channel between \ac{ul} and \ac{dl} is not reciprocal, but we can exploit the recently proposed distributional shift invariance of the \ac{ul} and \ac{dl} \cite{9714227}. To this end, the training process can be performed centralized at the \ac{bs} and the parameters are offloaded to the \ac{mt}, where the \ac{dl}-\ac{ce} is performed \cite{9739092}. However, the proposed approaches are not limited to this and can be equivalently used in the \ac{ul} and in a \ac{tdd} system.
	When only $ N_p $ positions of the $N_t \times N_c$ time-frequency response are occupied by pilot symbols, there is a \textit{selection matrix} $ \B A \in \{0,1\}^{N_p\times N_c N_t} $ which represents the pilot positions.
	\blue{This leads to observations as described in \eqref{eq:system_model} where $N=N_cN_t$ and $M=N_p$}.
	In this work, we consider a diamond-shaped pilot allocation scheme which is known to be \ac{mse}-optimal \cite{Choi2005}.
	
	For the construction of a scenario-specific channel dataset, we use the \quadriga~channel simulator \cite{QuaDRiGa1}.
	%	\quadriga models the channel of the $ c $-th carrier and $t$-th time symbol as
	%	$H_{c,t} = \sum_{\ell=1}^{L} G_{\ell} e^{-2\pi j f_c \tau_{\ell,t}}$,
	%	where \( \ell \) is the path number and $L$ the number of multi-path components.
	%	%depends on whether there is \ac{los}, \ac{nlos}, or \ac{o2i} propagation: \( L_\text{LOS} = 37 \), \( L_\text{NLOS} = 61 \) or \( L_\text{O2I} = 37 \).
	%	The frequency of the \( c \)-th carrier is denoted by \( f_c \) and the \( \ell\)-th path delay of the \( t \)-th time symbol by \( \tau_{\ell, t} \).
	%	The coefficient \( G_{\ell} \) comprises the attenuation of a path, the antenna radiation pattern weighting, and the polarization.
	We consider an \ac{uma} scenario following the 3GPP 38.901 specification, where the \ac{bs} is placed at a height of 25m and covers a sector of 120°.
	Each \ac{mt} is either placed indoors (80\%) or outdoors (20\%) and moves with a certain velocity $v$ in a random direction, which is captured by a drifting model.
	%	The generated channels are post-processed to remove the path gain~\cite{QuaDRiGa2}.

	\section{Learning a GMM from Imperfect Data}
	
	\subsection{Scenario-Specific Imperfect Training Dataset}\label{subsec:dataset}
	Having access to training data that represent the environment of a \ac{bs} cell, i.e., the typically complex and intractable channel \ac{pdf} $f_{\B h}$, in combination with \ac{ml} techniques, has been shown to substantially improve the channel estimation performance.
	The common prerequisite of these data-aided techniques is the availability of a training dataset $\mathcal{H}$ consisting of perfect \ac{csi} \blue{samples}, i.e., $L$ channel realizations %with fully-allocated pilots 
	$\mathcal{H} = \{\B h_\ell\}_{\ell=1}^L$. However, 
	this assumption is impractical if the training data are collected at a \ac{bs} during regular operation. That is, the training data stem from certain pilot positions at a finite \ac{snr}, i.e., $\mathcal{Y} = \{\B y_\ell\}_{\ell=1}^L$, cf. \eqref{eq:system_model}. Note that these data can be pre-processed in different ways, e.g., by denoising or interpolating the pilot positions. Although these training data are corrupted by various imperfections, the statistical information about the system model can be utilized to mitigate these effects, as shown in the following.
	%Nonetheless, the training data are corrupted by various imperfections.
	%In contrast, if the training data are generated by simulators, there always exists an unavoidable mismatch to the real distribution of channels that represents the environment of a specific \ac{bs} cell which cannot be accurately captured. 
	%To summarize, in this letter, we assume the availability of a training dataset $\mathcal{Y}$ purely comprised of pilot observations. Since we do not have access to ground-truth \ac{csi} labels, this introduces an unsupervised training task.

	\subsection{\blue{GMM-based} Channel Estimator with Adapted Training}\label{subsec:estimator}
	In \cite{9842343}, a channel estimator was introduced that is based on a \ac{gmm} whose application consists of two phases. First, in an offline training phase a $K$-components \ac{gmm} of the form 
	\begin{equation}\label{eq:gmm_h}
		f_{\B h}^{(K)}(\B h) = \sum_{k=1}^K \pi_k \mathcal{N}_\C(\B h; \B \mu_k, \B C_k),
	\end{equation}
	where $\pi_k$, $\B \mu_k$, and $\B C_k$ are the mixing coefficients, means, and covariances of the $k$th \ac{gmm} component, respectively, is fitted by maximum likelihood optimization with the well-known \ac{em} algorithm, cf. \cite[Sec. 9.2]{bookBi06}, via  training samples from $\mathcal{H}$ in order to approximate the underlying distribution $f_{\B h}$ of the channels in the whole \ac{bs} cell. \blue{Due to the Gaussianity of the noise with the known covariance matrix $\B C_{\B n}$ and the known observation matrix $\B A$, the \ac{pdf} of the observations $f_{\B y}$ can be approximated by means of the \ac{gmm} $f_{\B y}^{(K)}(\B y) = \sum_{k=1}^K \pi_k \mathcal{N}_\C(\B y; \B A \B \mu_k, \B C_{\B y,k})$, where $\B C_{\B y,k} = \B A\B C_k\B A\h + \B C_{\B n}$, which can be straightforwardly computed once \eqref{eq:gmm_h} is known, cf. \eqref{eq:system_model}.}
	Second, the trained \ac{gmm} is leveraged to perform channel estimation by computing a convex combination of linear \ac{mmse} estimates for a given pilot signal $\B y$, i.e.,
	\begin{equation}
		\hat{\B h}_{\text{GMM}} = \sum_{k=1}^K \gamma_k(\B y) \left( \B C_k\B A\h\B C_{\B y,k}\inv (\B y - \B A \B \mu_k) + \B \mu_k\right)
	\end{equation}
	where $\gamma_k(\B y)$ is the responsibility of the $k$th component\blue{, i.e., the probability that the component $k$ is responsible for the observation $\B y$ \cite[Sec. 9.2]{bookBi06}.}
	In \cite{gmm_structured}, the estimator was extended to impose structural features to the channel covariances $\B C_k$ in the training of the \ac{gmm} for saving complexity and memory.
	
	%In \cite{gmm_structured,9940363}, this was extended to structured covariances and applied to measurement data, respectively. 
	
	However, so far, the \ac{gmm} was trained on a dataset comprised of perfect channels $\mathcal{H}$, cf. \Cref{subsec:dataset}. If we naively replace $\mathcal{H}$ with the noisy, possibly pre-processed, training data $\mathcal{Y}$ from pilot observations \eqref{eq:system_model}, a severe performance loss can be expected due to the imperfections in the training data.
	Therefore, in the following, we propose adapted training procedures for the \ac{gmm} that drastically alleviates the effects of the imperfect training data by utilizing the knowledge of the system model \eqref{eq:system_model} and possibly existing structural features of the covariances.

	%\subsubsection{Training Data Corrupted by AWGN}\label{subsub:awgn}
	
	We first discuss the case of training data that is corrupted by \ac{awgn} following the model in \eqref{eq:system_model} with $\B A = \eye_{\blue{N}}$, cf. \Cref{subsec:spatial_channel_model}. The main idea is to adapt the \ac{em} algorithm such that the \ac{pdf} $f_{\B y}$ of the observations is approximated by the \ac{gmm} $f_{\B y}^{(K)}(\B y)$.
	%\begin{equation}\label{eq:gmm_y}
	%		f_{\B y}^{(K)}(\B y) = \sum_{k=1}^K p(k) \mathcal{N}_\C(\B y; \B \mu_k, \B C_{\B y,k}).
	%\end{equation}
	Due to zero-mean \ac{awgn}, the updates of the mixing coefficients $\pi_k$ and the means $\B \mu_k$ are unchanged with respect to the classical \ac{em} algorithm, cf. \cite[Sec. 9.2]{bookBi06}.
	However, 
	%the covariance $\B C_{\B y,k}$ of the $k$th component can be any \ac{psd} matrix in the classical \ac{em} algorithm, although 
	we would like to include the constraint that $\B C_{\B y,k} = \B C_k + \B C_{\B n}$, following the model \eqref{eq:system_model}.
	The \ac{gmm} for the channel distribution \eqref{eq:gmm_h} can then be directly obtained since $\B C_{\B n}$ is known.
	The derivation for the update of $\B C_k$ is given in the following.
	
	\begin{theorem}\label{theorem:awgn}
		Given noisy pilot observations $\mathcal{Y}$, the maximum likelihood solution $\B C_k^*$ for the covariance in the \ac{em} algorithm is given by computing the \ac{evd} 
		\begin{equation}\label{eq:evd}
			\hat{\B C}_{\B y,k} - \B C_{\B n} = \B V \diag(\B \xi)\B V\h
		\end{equation}
		with $\hat{\B C}_{\B y,k} =\frac{1}{N_k} \sum_{\ell=1}^L \gamma_{k,\ell} (\B y_\ell - \B\mu_k)(\B y_\ell - \B\mu_k)\h$, $N_k =  \sum_{\ell=1}^L \gamma_{k,\ell}$, and $\gamma_{k,\ell}$ being the responsibility of component $k$ for data point $\B y_\ell$. Afterwards, an elementwise truncation of the negative eigenvalues via $\B \xi^{\text{PSD}} = \max(\B 0, \B \xi)$ is performed such that
		\begin{equation}\label{eq:psd_solution}
			\B C_{k}^* =  \B V \diag(\B \xi^{\text{PSD}})\B V\h.
		\end{equation}
	\end{theorem}
	
	\begin{proof}
		The maximization of the expected complete log-likelihood, cf. \cite[Sec. 9.3]{bookBi06}, of the pilot observations with respect to the $k$th channel covariance $\B C_k$ is given by
		\blue{
			\begin{align}
				\B C_k^* &=  \argmax_{\B C_k\succeq\B 0} \sum_{j=1}^K\sum_{\ell=1}^L \gamma_{j,\ell} \log \mathcal{N}_{\mathbb{C}}(\B y_\ell; \B \mu_j , \B C_j+\B C_{\B n})
				\label{eq:proof11}
				\\
				&=  \argmax_{\B C_k\succeq\B 0} \sum_{\ell=1}^L \gamma_{k,\ell} \log \mathcal{N}_{\mathbb{C}}(\B y_\ell; \B \mu_k , \B C_k+\B C_{\B n})
				\label{eq:proof12}
				\\
				&=   \argmax_{\B C_k\succeq\B 0} 
				-\log\det (\B C_k + \B C_{\B n})
				- \tr(\hat{\B C}_{\B y,k} (\B C_k + \B C_{\B n})\inv)
				\label{eq:proof13}
				%	\\
				%	&=   \argmax_{\B C_k\succcurlyeq\B 0} \sum_{\ell=1}^L \gamma_{k,\ell} \left(
				%	-\log\det (\B C_k + \B C_{\B n})\right. 
				%	\\ 
				%	&\hspace*{3cm} \left. - \tr(\hat{\B C}_{\B y,k} (\B C_k + \B C_{\B n})\inv)\nonumber
				%	\right)
			\end{align}
			where from \eqref{eq:proof11} to \eqref{eq:proof12} it is used that the maximization depends only on the $k$th \ac{gmm} component, and \eqref{eq:proof13} follows from the definition of the Gaussian \ac{pdf} and $\hat{\B C}_{\B y,k}$ from above.
			Thus, the optimization simplifies to a maximization of the Gaussian likelihood of the $k$th \ac{gmm} component.
		}
		%	with $\hat{\B C}_{\B y,k} =\frac{1}{N_k} \sum_{\ell=1}^L \gamma_{k,\ell} (\B y_\ell - \B\mu_k)(\B y_\ell - \B\mu_k)\h$ and $N_k =  \sum_{\ell=1}^L \gamma_{k,\ell}$. 
		%	Note that this sample covariance matrix has the form of the usual covariance update in the \ac{em} without \ac{awgn}.
		The optimal solution for the case of a Gaussian likelihood is derived in \cite[Appendix]{7051818} and
		is obtained with the \ac{evd} and the truncation of the negative eigenvalues.
		This can be applied \blue{to \eqref{eq:proof13}}
		through a substitution of 
		%	the weighted sample covariance 
		$\hat{\B C}_{\B y,k}$ 
		which yields the solution in \eqref{eq:psd_solution}. 
		%	The optimal solution to this simplified problem is derived in \cite[Appendix]{7051818} which is computed by performing the \ac{evd} \eqref{eq:evd} and projecting the negative eigenvalues to zero which yields the solution in \eqref{eq:psd_solution}.
	\end{proof}

	%Thus, in order to fulfill this additional constraint, we adapt the \ac{em} algorithm in the following way. We set $\B C_{k} = \B C_{\B y,k} - \B C_{\B n}$ with fixed $\B C_{\B n}$ and perform the updates of the covariances in the M-step only with respect to the channel covariances $\B C_k$. To enforce the \ac{psd} constraint, we implement a projection step via an \ac{evd}. The derivation of the update step is shown in the Appendix.  
	
	%Note that the update of the covariance in the classical \ac{em} algorithm without \ac{awgn} has exactly the form of $\hat{\B C}_{\B y,k}$. 

	%\subsubsection{Training Data with Missing Entries}\label{subsub:ofdm}
	Next, we consider the case where the pilot observations stem from the \ac{ofdm} model in \Cref{subsec:doubly_channel_model}, i.e., only a few channel entries are observed according to the pilot pattern represented by $\B A$.
	% the case of training data that stem from noisy pilot observations where only a few channel entries are observed according to a specific \ac{ofdm} pilot pattern, cf. \Cref{subsec:doubly_channel_model}, is discussed. 
	%For the general case of noiseless data there exists an adaptation of the \ac{em} algorithm for missing entries where in each iteration, the datapoints are interpolated for each component in the \ac{gmm} according to the current estimates of the parameters \cite[Chapter 11]{bookRubin02}. 
	To adapt the training procedure, we combine the insights from Theorem \ref{theorem:awgn} together with the \ac{em} algorithm for missing entries, cf. \cite[Chapter 11]{bookRubin02}. 
	We first define a selection matrix $\bar{\B A}\blue{\in\{0,1\}^{(N_cN_t-N_p)\times N_cN_t}}$ which represents the positions that are not allocated with pilots, i.e., $\B A \bar{\B A}\T = \B 0$. 
	After an initial interpolation step, one \ac{em} iteration is performed in the usual way to have initial estimates of the \ac{gmm} parameters. 
	\blue{
		The updates of the mixing coefficients $\pi_k$ and means $\B \mu_k$ are again unchanged with respect to the classical \ac{em} algorithm due to the zero-mean \ac{awgn}.}
	The adapted updates of the covariances $\B C_k$ are derived in the following.
	
	\begin{theorem}\label{theorem:missing}
		%	Let $\bar{\B A}$ denote a selection matrix which represents the positions that are not allocated with pilots, i.e., $\B A \bar{\B A}\T = \B 0$. 
		%	After an initialization of the \ac{gmm} parameters, e.g., by linear interpolation,
		Given sparsely allocated and noisy pilot observations $\mathcal{Y}$, first, a linear \ac{mmse} estimate of the unobserved and noisy channel entries $\bar{\B y }_{\ell,k}$ is computed via the current statistics of the $k$th \ac{gmm} component as
		\begin{equation}\label{eq:interpolation}
			\bar{\B y }_{k,\ell} = \bar{\B A} \B \mu_k + \bar{\B A} \B C_k\B A\T (\B A\B C_k\B A\T + \B C_{\B n})\inv (\B y_\ell - \B A \B \mu_k).
		\end{equation}
		Afterwards, a fully interpolated sample, given as $\hat{\B y }_{k,\ell} = \B A\T \B y_\ell + \bar{\B A}\T \bar{\B y }_{k,\ell}$, is used to update
		\begin{equation}\label{eq:update_Cy}
			\hat{\B C}_{\B y,k} =\frac{1}{N_k} \sum_{\ell=1}^L \gamma_{k,\ell} (\hat{\B y }_{k,\ell} - \B\mu_k)(\hat{\B y }_{k,\ell} - \B\mu_k)\h +\bar{\B A}\T \B \Sigma_k \bar{\B A}
		\end{equation}
		where $N_k = \sum_{\ell=1}^L \gamma_{k,\ell} $ and 
		%$\B \Sigma_k = \B C_k^{mm} - \B C_k^{mo}(\B C_k^{oo} + \B C_{\B n})\inv \B C_k^{om}$.
		\begin{equation}\label{eq:cov_unobserved}
			\B \Sigma_k = \bar{\B A} \B C_k \bar{\B A}\T - \bar{\B A} \B C_k\B A\T(\B A\B C_k\B A\T + \B C_{\B n})\inv \B A\B C_k\bar{\B A}\T.
		\end{equation}
		Finally, to account for the \ac{awgn}, the update of the channel covariance matrix \blue{$\B C_k^*$} is computed via the \ac{evd} of $\hat{\B C}_{\B y,k} - \B C_{\B n}$ and the truncation of the negative eigenvalues, cf. Theorem \ref{theorem:awgn}.
	\end{theorem}
	
	\begin{proof}
		The steps \eqref{eq:interpolation}--\eqref{eq:cov_unobserved} are derived in the \ac{em} algorithm for missing data, cf. \cite[Chapter 11]{bookRubin02}, where the additional covariance term \eqref{eq:cov_unobserved} accounts for the estimated covariance of the missing entries. 
		\blue{Given the estimate of the observation covariance $\hat{\B C}_{\B y,k}$ in \eqref{eq:update_Cy}, the subsequent projection via the \ac{evd} to get the maximum likelihood estimation of the channel covariance $\B C_k^*$ is a direct consequence of Theorem \ref{theorem:awgn}.}
		%	The subsequent projection via the \ac{evd} is a direct consequence of Theorem \ref{theorem:awgn}.
	\end{proof}

	Additionally, it is possible to enforce covariances with a specific structure \cite{gmm_structured}. In this letter, we are focusing especially on the case of block-Toeplitz covariances, constructed as 
	$\B C_k = \B Q\h \diag(\B c_k)\B Q$,
	where $\B Q$ is a truncated 2D-\ac{dft} matrix. 
	This structure fits particularly well in the \ac{ofdm} case \cite{gmm_structured}. It is important to mention that the imposed structural constraints on the covariances can be understood as a regularization technique. Interestingly, as shown later in \Cref{sec:sim_results}, this regularization allows for performance gains, especially in the case of sparse pilot allocations.
	
	Algorithm \ref{alg:adapted_em_structured} summarizes the proposed adapted \ac{em} algorithm for fitting the \ac{gmm} in the \ac{ofdm} case from \Cref{subsec:doubly_channel_model} with noisy and sparsely allocated pilot observations with structured covariances. If we consider the spatial model from \Cref{subsec:spatial_channel_model}, the steps 8, 9, and 14 can be omitted. For unstructured covariances, the steps 2, 20, 21, and 22 are dropped.

	%observed and missing part of a vector/matrix, respectively. In the M-step of the \ac{em} algorithm, the datapoints from $\mathcal{Y}$ are first interpolated for the $k$th \ac{gmm} component as
	%\begin{equation}
	%	\hat{\B y }_\ell^{k,m} = \B \mu_k^m + \B C_k^{mo} (\B C_k^{oo} + \B C_{\B n})\inv (\B y_\ell - \B \mu_k^o) 
	%\end{equation}
	%where $\B C_k$ and $\B \mu_k$ are the current parameters of the $k$th component.
	
	\subsection{Discussion about Robustness}\label{subsec:robustness}
	
	%After the investigation of an adapted training strategy for the \ac{gmm} estimator such that it can be trained in an unsupervised manner, we discuss the advantages of the proposed method in terms of its robustness in the following.
	
	The first stage of training the \ac{gmm} is equivalent to learn a generative model that represents the channel distribution of the whole \ac{bs} cell.
	On the one hand, this generative model allows to leverage prior information about the distribution of the channels in the whole \ac{bs} cell that enhances the channel estimation performance \cite{9842343}. On the other hand, it can be adapted to imperfections in the training data as discussed above which is motivated by model-based insights. This is a fundamental difference to common learning-based estimators that are trained to learn a regression mapping between pilot observations and channel estimates, which inherently rely on perfect \ac{csi} \blue{samples}, e.g., \cite{NeWiUt18,8640815,8052521}. Thus, the proposed \ac{gmm} approach allows for a more robust solution with respect to various imperfections in the training data because of its ability to adapt the training procedure accordingly and introduce structural regularization, cf. \Cref{subsec:estimator}. 
	%This is confirmed by simulations in \Cref{sec:sim_results}.
	
	\begin{algorithm}[t]
		\caption{Adapted \ac{em} with Structured Covariances.}
		\label{alg:adapted_em_structured}
		\begin{algorithmic}[1]
			%		\renewcommand{\algorithmicensure}{\textbf{Offline GMM Training Phase}}
			%		\ENSURE
			\REQUIRE $\mathcal{Y}$, $K$, $\B A$, $\B C_{\B n}$, $\{\B\mu_k^{(1)}, \B C_k^{(1)}, \pi_k^{(1)}\}_{k=1}^K$, $\B Q$, $i=1$, $i_{\text{max}}$
			\vspace{-0.35cm}
			\STATE Get selection matrix $\bar{\B A}$ such that $\B A\bar{\B A}\T = \B 0$
			\STATE Initialize $\{\B c_k^{(1)} \leftarrow \B Q \B C_k^{(1)} \B Q\h\}_{k=1}^K$
			%		\STATE Set $i=1$
			\renewcommand{\algorithmicendfor}{\textbf{end}}
			\renewcommand{\algorithmicendwhile}{\textbf{end}}
			\WHILE{$i  <i_{\text{max}}$ and convergence criterion not met}
			\STATE $\{\B C_{\B y,k}^{(i)} \leftarrow  \B A \B C_k^{(i)}\B A\T + \B C_{\B n}\}_{k=1}^K$
			\FOR{$k=1$ to $K$}
			\FOR {$\ell =1$ to $L$}
			\STATE $\gamma_{k,\ell} \leftarrow \frac{\pi_k^{(i)} \mathcal{N}_\C(\B y_{\ell}; \B A\B \mu_k^{(i)},\B C_{\B y,k}^{(i)})}{\sum_{j=1}^K \pi_j^{(i)} \mathcal{N}_\C(\B y_{\ell}; \B A\B \mu_j^{(i)}, \B C_{\B y,j}^{(i)})}$
			\COMMENT{E-step}
			\STATE 	$\bar{\B y }_{k,\ell} \leftarrow \bar{\B A} \B \mu_k^{(i)} + \bar{\B A} \B C_k^{(i)}\B A\T \B C_{\B y,k}^{(i),-1} (\B y_\ell - \B A \B \mu_k^{(i)})$
			\STATE $\hat{\B y }_{k,\ell} \leftarrow \B A\T \B y_\ell + \bar{\B A}\T \bar{\B y }_{k,\ell}$
			\COMMENT{interpolated sample}
			\ENDFOR
			\STATE $N_k \leftarrow \sum_{\ell =1}^L \gamma_{k,\ell}$
			\STATE $\pi_k^{(i+1)} \leftarrow \frac{N_k}{L}$
			\COMMENT{mixing coefficient update}
			\STATE $\B\mu_k^{(i+1)} \leftarrow \frac{1}{N_k} \sum_{\ell=1}^L\gamma_{k,\ell} \hat{\B y }_{k,\ell} $
			\COMMENT{mean update}
			\STATE $\B \Sigma_k \leftarrow \bar{\B A} \B C_k^{(i)} \bar{\B A}\T - \bar{\B A} \B C_k^{(i)}\B A\T\B C_{\B y,k}^{(i),-1} \B A\B C_k^{(i)}\bar{\B A}\T$
			\STATE  $\hat{\B y }_{k,\ell} \leftarrow \hat{\B y }_{k,\ell}  - \B\mu_k^{(i+1)}$
			\STATE $\B C_{\B y,k}^{(i+1)} \leftarrow \frac{1}{N_k} \sum\limits_{\ell=1}^L \gamma_{k,\ell}\hat{\B y }_{k,\ell}\hat{\B y }_{k,\ell}\h +\bar{\B A}\T \B \Sigma_k \bar{\B A}$
			\STATE $\B V_k,~\B \xi_k \leftarrow \operatorname{EVD}(\B C_{\B y,k}^{(i+1)} - \B C_{\B n})$
			\STATE $\B \xi_k^{\text{PSD}} \leftarrow \max(\B 0,  \B \xi_k)$ 
			\COMMENT{elementwise max} 
			\STATE $\B C_k^{(i+1)} \leftarrow \B V_k \diag(\B \xi_k^{\text{PSD}})\B V_k\h$
			\STATE $\B \Theta_k \leftarrow \B Q \left(\B C_k^{(i),-1}\B C_k^{(i+1)}  \B C_k^{(i),-1} - \B C_k^{(i),-1}\right)\B Q\h$
			\STATE $\B c_k^{(i+1)} \leftarrow \B c_k^{(i)} + \diag\left(\diag(\B c_k^{(i)}) \B \Theta_k \diag(\B c_k^{(i)}) \right)$
			\STATE $\B C_k^{(i+1)} \leftarrow \B Q\h \diag(\B c_k^{(i+1)}) \B Q$
			\COMMENT{covariance update}
			\ENDFOR
			\STATE $i\leftarrow i +1 $
			\ENDWHILE		
			\RETURN $\{\B \mu_k^{(i)}, \B C_k^{(i)}, \B \pi_k^{(i)} \}_{k=1}^K$
		\end{algorithmic}
		%	\vspace{-0.01cm}
	\end{algorithm}
	
	%This results in a more robust solution with respect to various imperfections in the training data.
	
	An important advantage of the \ac{gmm} estimator is that only the training phase must be changed in contrast to \cite{9842343} where perfect training \ac{csi} is assumed to be available. Thus, the online complexity and the number of parameters of the estimator do not change.
	Furthermore, the \ac{gmm} is universal and requires further adaptation only in case of changing parameters of the propagation environment which can be conveniently tracked with the frequently received pilots at the \ac{bs}.
	Additionally, it has to be trained only once for a given \ac{snr} and can then be applied to any other \ac{snr} value, which is in contrast to learning-based estimators where the dependency on the \ac{snr} is crucial, e.g., \cite{8640815,NeWiUt18}.
	Fortunately, since the proposed covariance updates in Theorem \ref{theorem:awgn} and Theorem \ref{theorem:missing} are the maximum likelihood solutions, all favorable properties of the \ac{em} algorithm, such as a monotonically increasing likelihood, are preserved. 
	
	\section{Simulation Results}\label{sec:sim_results}
	We present numerical results to evaluate the proposed method for two different channel models in comparison to state-of-the-art estimators. We set $\op E[\|\B h\|^2]=N$ such that the $\text{SNR} = 1/\sigma^2$. We choose $K=64$ \ac{gmm} components \blue{which is practically reasonable \cite{9842343}}. We utilize $L=10^5$ training samples \blue{for all data-based approaches}, and evaluate on $10^4$ test samples \blue{of different \acp{mt} that are not part of the training data}.
	%We choose $K=64$ components for the \ac{gmm} estimator.
	
	%For both discussed scenarios we compare to data-aided state-of-the-art \ac{ml} approaches. It is important to note that 
	
	\subsection{Spatial Channel Model}
	
	We first evaluate the channel estimation performance for the spatial model from \Cref{subsec:spatial_channel_model} with $\B A = \eye_{\blue{N}}$. 
	The \ac{gmm} estimator which is based on perfect \ac{csi} is denoted by ``GMM $\mathcal{H}$'', whereas the \ac{gmm} estimator that is naively used without any modifications on training data from $\mathcal{Y}$ is denoted by ``GMM mismatch''. The proposed approach with adapted training from \Cref{subsec:estimator} is denoted by ``GMM $\mathcal{Y}$''.
	We analyze the performance in comparison with the following baseline estimators.
	The curve labeled ``genie'' represents the utopian \ac{mmse} estimator that has full knowledge of $\B C_{\B\delta}$ for each sample. We further evaluate the \ac{ls} solution $\hat{\B h}_{\text{LS}} = \B y$. 
	%	Note that these approaches do not need any training data.
	We include the \ac{cnn} estimator from \cite{NeWiUt18}, that learns a regression mapping, labeled ``CNN''. \blue{The \ac{cnn} consists of two layers with \ac{relu} activation and we use the truncated \ac{dft} matrix for the input/output transform, cf. \cite{NeWiUt18}}. To achieve a fair comparison, the \ac{cnn} estimator is trained on data from $\mathcal{Y}$, which clearly introduces an unavoidable mismatch in the learning phase.
	%	Lastly, we evaluate the estimator from \cite{vae_chEst}, where a \ac{vae} is fitted onto the training data as a generative model and afterwards leveraged for channel estimation, similar to the \ac{gmm} estimator.  
	%	We choose the version which is trained unsupervised on $\mathcal{Y}$ with a modified training, labeled ``VAE-real'', cf. \cite{vae_chEst}.
	
	In Fig. \ref{fig:mse_simo_1path}, the channel estimation performance of the above discussed estimators is shown for $N=128$ \ac{bs} antennas and one propagation cluster, consisting of multiple sub-paths. 
	The depicted \ac{snr} is the same in both training and evaluation, i.e., a different set of training data $\mathcal{Y}$ is given for each \ac{snr} value, which mimics a realistic situation.
	The \ac{gmm} estimator with perfect \ac{csi} $\mathcal{H}$ performs very close to the genie-\ac{mmse} estimator which is in accordance with the findings in \cite{9842343}. If the \ac{gmm} estimator is fitted naively with data from $\mathcal{Y}$ without modifications to the training procedure, this leads to a severe performance loss of about $5$dB for the whole \ac{snr} range. However, with the proposed modifications, a performance close to the perfect \ac{csi} case, and thus to the genie-\ac{mmse} estimator, is possible, purely based on training data from $\mathcal{Y}$.  
	The \ac{cnn} estimator has a substantial performance loss as compared to the proposed estimator, especially in the high \ac{snr} regime. 
	%	In contrast, the \ac{vae} estimator also performs well, which is in accordance with the discussed properties of generative models that allow for adaptations to the system model, cf. \Cref{subsec:robustness}.
	%	Note that a separate forward pass through the \ac{vae} has to performed for each estimate, leading to possibly higher complexity.
	\begin{figure}[t]
	\centering
	\begin{tikzpicture}
		\begin{axis}
			[width=\plotwidth,
			height=\plotheight,
			xtick=data,
			ytick={1,1e-1,1e-2,1e-3,1e-4},
			xmin=-10, 
			xmax=30,
			xlabel={SNR [dB]},
			ymin= 1e-4,
			ymax= 1,
			ymode=log,
			label style={font=\small},
			tick label style={font=\small},
			ylabel= {Normalized MSE}, 
			ylabel shift = 0.0cm,
			grid = both,
			legend columns = 2,
			legend entries={
				 genie,
				 LS,
				 CNN,
				 GMM $\mathcal{H}$,
				 GMM mismatch,
				 GMM $\mathcal{Y}$,
%				\scriptsize VAE-real,
			},
			legend style={at={(0.0,0.0)}, anchor=south west, font=\scriptsize},
			]
			
			\addplot[genie]
			table[x=SNR, y=genie, col sep=comma]
			{csvdat/2022-09-28_18-49-18_n_paths=1_ant=128_comp=128_3gpp=True_full.csv};
			
			\addplot[LS]
			table[x=SNR, y=LS, col sep=comma]
			{csvdat/2022-11-22_16-31-27_n_paths=1_ant=128_comp=128_3gpp=True_full.csv};
			
			\addplot[cnn1]
			table[x=SNR, y=cnn_fft2x_relu_non_hier_False, col sep=comma]
			{csvdat/2022-11-23_11-03-12_1paths_128antennas_20lbatch_20lsize_500ebatch_snr_train=perSNR_noisy.csv};
			
			\addplot[gmm]
			table[x=SNR, y=gmm, col sep=comma]
			{csvdat/2022-12-13_16-30-56_n_paths=1_ant=128_comp=64_3gpp=True_full.csv};
			
			\addplot[gmm_perSNR]
			table[x=SNR, y=gmm_perSNR_mismatch, col sep=comma]
			{csvdat/2022-09-28_18-51-54_n_paths=1_ant=128_comp=64_3gpp=True_full_noisyGMM.csv};

			\addplot[gmm_20dB]
			table[x=SNR, y=gmm_noisy, col sep=comma]
			{csvdat/2022-09-28_18-51-54_n_paths=1_ant=128_comp=64_3gpp=True_full_noisyGMM.csv};
			
%			\addplot[vae]
%			table[x=SNR, y=vae_real_perSNR, col sep=comma]
%			{csvdat/vae_circ_noisy_test-perSNR.csv};
			
		\end{axis}
	\end{tikzpicture}
	\caption{Spatial channel model from \Cref{subsec:spatial_channel_model} with $N=128$ \ac{bs} antennas and one propagation cluster. The \ac{snr} is the same for training and evaluation.}
	\label{fig:mse_simo_1path}
\end{figure}
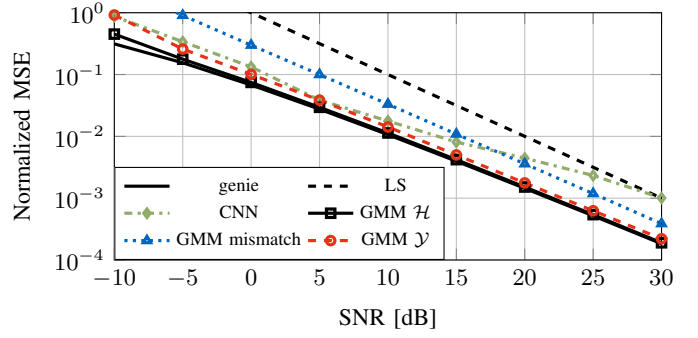

	\subsection{Doubly-Selective Fading Channel Model}

	We consider a typical \ac{ofdm} frame structure of $N_c = 12$ carriers having a spacing of $15$kHz, yielding a bandwidth of $180$kHz, and $N_t=14$ time slots, cf. \Cref{subsec:doubly_channel_model}. \blue{In both the training and the test dataset we consider \acp{mt} which move with a fixed but random velocity between three and $130$km/h, i.e., each \ac{mt} in the datasets has a velocity drawn from the uniform distribution $v\sim\mathcal{U}(3,130)$km/h. This should represent a practical \ac{bs} environment where different \acp{mt} are moving with a different velocity.}
	We evaluate once again the \ac{gmm} estimator with perfect \ac{csi} from $\mathcal{H}$ and a naive version where the \ac{gmm} is fitted with linearly interpolated data from $\mathcal{Y}$, labeled ``GMM lin-int''. The proposed approach with adapted training from \Cref{subsec:estimator} and linear interpolation as initialization is denoted by ``GMM $\mathcal{Y}$''. In the \ac{ofdm} case, we can impose structural features to the covariances as regularization on the adapted approach. In this letter, we focus on block-Toeplitz matrices as discussed in \cite{gmm_structured}, labeled as ``GMM $\mathcal{Y}$ toep''.
	We compare our proposed methods with the typically used linear interpolator, labeled ``lin-int'', and with a linear \ac{mmse} estimator based on a cell-wide global sample covariance computed with linearly interpolated data from $\mathcal{Y}$, labeled as ``samp-cov lin-int''. In this case the genie-\ac{mmse} estimator is not computable since no covariance statistics are provided from the simulator \cite{QuaDRiGa1}.
	We also compare with the \ac{ml}-based \textit{ChannelNet} from \cite{8640815} \blue{(we adopt the same architecture and hyperparameters) where the training data are comprised of linearly interpolated samples} from $\mathcal{Y}$ for a fair comparison.

\begin{figure}[t]
	\centering
	\begin{tikzpicture}
		\begin{axis}
			[width=\plotwidth,
			height=\plotheight,
			xtick=data,
			ytick={1,1e-1,1e-2,1e-3,1e-4},
			xmin=-10, 
			xmax=30,
			xlabel={SNR [dB]},
			ymin= 1e-4,
			ymax= 1,
			ymode=log,
			ylabel= {Normalized MSE}, 
			label style={font=\small},
			tick label style={font=\small},
			ylabel shift = 0.0cm,
			grid = both,
			legend columns = 2,
			legend entries={
				%\scriptsize LS,
%				\scriptsize genie-OMP,
				 lin-int,
				 GMM $\mathcal{H}$,
				 samp-cov lin-int,
				 GMM lin-int,
				 GMM $\mathcal{Y}$,
				 GMM $\mathcal{Y}$ toep,
				 ChannelNet,
			},
			legend style={at={(0.0,0.0)}, anchor=south west, font=\scriptsize},
			]
			
			%\addplot[mark options={solid},color=black,line width=1.2pt,mark=square]
			%table[x=SNR, y=ls, col sep=comma]
			%{csvdat/2022-10-21_11-07-24_carrier=24_symbols=14_pilots=50_pattern=lattice_sum=0.99_ntrain=100000_covtype=block-toeplitz_v=3kmh_noisy.csv};
			
%			\addplot[omp]
%			table[x=SNR, y=omp, col sep=comma]
%			{csvdat/2022-12-13_15-26-55_carrier=12_symbols=14_pilots=18_pattern=diamond_sum=0.99_ntrain=100000_covtype=full_v=130kmh_noisy.csv};
			
			\addplot[LS]
			table[x=SNR, y=lin_int, col sep=comma]
			{csvdat/2022-12-10_11-23-01_carrier=12_symbols=14_pilots=18_pattern=diamond_sum=0.99_ntrain=100000_v=130kmh.csv};
			
			\addplot[gmm]
			table[x=SNR, y=gmm_full, col sep=comma]
			{csvdat/2022-12-10_13-16-57_carrier=12_symbols=14_pilots=18_pattern=diamond_sum=0.99_ntrain=100000_v=130kmh.csv};

			\addplot[samplecov]
			table[x=SNR, y=lin_int_sample_cov_perSNR, col sep=comma]
			{csvdat/2022-12-07_16-39-47_carr=12_symb=14_pilots=18_sum=0.99_ntrain=100000_covtype=full_trainpatt=diamond_pattern=diamond_v=130kmh_comps=64_noisy_perSNR.csv};
			
			\addplot[gmm_perSNR]
			table[x=SNR, y=lin_int_gmm_mismatch_perSNR, col sep=comma]
			{csvdat/2022-12-07_16-39-47_carr=12_symb=14_pilots=18_sum=0.99_ntrain=100000_covtype=full_trainpatt=diamond_pattern=diamond_v=130kmh_comps=64_noisy_perSNR.csv};
			
			\addplot[gmm_20dB]
			table[x=SNR, y=gmm_full_perSNR, col sep=comma]
			{csvdat/2022-12-07_16-39-47_carr=12_symb=14_pilots=18_sum=0.99_ntrain=100000_covtype=full_trainpatt=diamond_pattern=diamond_v=130kmh_comps=64_noisy_perSNR.csv};

			\addplot[gmm_toep]
			table[x=SNR, y=gmm_block-toeplitz_perSNR, col sep=comma]
			{csvdat/2022-12-07_16-39-47_carr=12_symb=14_pilots=18_sum=0.99_ntrain=100000_covtype=full_trainpatt=diamond_pattern=diamond_v=130kmh_comps=64_noisy_perSNR.csv};
			
			\addplot[cnn1]
			table[x=SNR, y=channelnet, col sep=comma]
			{csvdat/2022-12-21_15-54-42_carr=12_symb=14_pilots=18_ntrain=100000_trainpatt=diamond_pattern=diamond_v=130kmh_perSNR.csv};

		\end{axis}
	\end{tikzpicture}
	\caption{\ac{ofdm} model from \Cref{subsec:doubly_channel_model} with $N_c=12$, $N_t=14$, and $N_p=18$. 
		The \ac{snr} is the same for training and evaluation.}
	\label{fig:mse_quadriga_perSNR}
\end{figure}
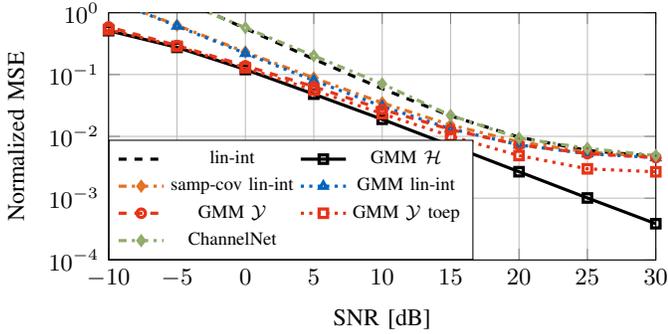

	\begin{figure}[t]
	\centering
	\begin{tikzpicture}
		\begin{axis}
			[width=\plotwidth,
			height=\plotheight,
			xtick=data,
			xmin=3, 
			xmax=34,
			xlabel={Pilots $N_p$},
			ymin= 4*1e-3,
			ymax= 2*1e-1,
			ymode=log,
			ylabel= {Normalized MSE}, 
			label style={font=\small},
			tick label style={font=\small},
			ylabel shift = 0.0cm,
			grid = both,
			legend columns = 2,
			legend entries={
				%				\scriptsize LS,
				%				\scriptsize genie-OMP,
				lin-int,
				GMM $\mathcal{H}$,
				samp-cov lin-int,
				GMM lin-int,
				GMM $\mathcal{Y}$,
				GMM $\mathcal{Y}$ toep,
				ChannelNet,
			},
			legend style={at={(1.0,1.0)}, anchor=north east, font=\scriptsize},
			]
			
			%			
			%			\addplot[mark options={solid},color=TUMBeamerGreen,line width=1.2pt,mark=x]
			%			table[x=SNR, y=omp, col sep=comma]
			%			{csvdat/2022-12-07_13-13-01_carrier=12_symbols=14_pilots=18_pattern=diamond_sum=0.99_ntrain=100000_covtype=full_v=200kmh_noisy.csv};
			%			
			
			\addplot[LS]
			table[x=pilots, y=lin_int, col sep=comma]
			{csvdat/2022-12-10_13-16-57_carrier=12_symbols=14_pattern=diamond_sum=0.99_ntrain=100000_v=130kmh_SNR=15dB_pilots.csv};

			\addplot[gmm]
			table[x=pilots, y=gmm_full, col sep=comma]
			{csvdat/2022-12-10_13-16-57_carrier=12_symbols=14_pattern=diamond_sum=0.99_ntrain=100000_v=130kmh_SNR=15dB_pilots.csv};

			\addplot[samplecov]
			table[x=pilots, y=lin_int_sample_cov_15dB, col sep=comma]
			{csvdat/2022-12-10_13-16-57_carrier=12_symbols=14_pattern=diamond_sum=0.99_ntrain=100000_v=130kmh_SNR=15dB_pilots.csv};
			\addplot[gmm_perSNR]
			table[x=pilots, y=lin_int_gmm_mismatch_full_15dB, col sep=comma]
			{csvdat/2022-12-10_13-16-57_carrier=12_symbols=14_pattern=diamond_sum=0.99_ntrain=100000_v=130kmh_SNR=15dB_pilots.csv};
			\addplot[gmm_20dB]
			table[x=pilots, y=gmm_full_15, col sep=comma]
			{csvdat/2022-12-10_13-16-57_carrier=12_symbols=14_pattern=diamond_sum=0.99_ntrain=100000_v=130kmh_SNR=15dB_pilots.csv};
			\addplot[gmm_toep]
			table[x=pilots, y=gmm_block-toeplitz_15, col sep=comma]
			{csvdat/2022-12-10_13-16-57_carrier=12_symbols=14_pattern=diamond_sum=0.99_ntrain=100000_v=130kmh_SNR=15dB_pilots.csv};
			\addplot[cnn1]
			table[x=pilots, y=channelnet, col sep=comma]
			{csvdat/2022-12-10_13-16-57_carrier=12_symbols=14_pattern=diamond_sum=0.99_ntrain=100000_v=130kmh_SNR=15dB_pilots.csv};

		\end{axis}
	\end{tikzpicture}
	\caption{\ac{ofdm} model from \Cref{subsec:doubly_channel_model} with $N_c=12$, $N_t=14$, and $\text{SNR} = 15$dB. 
		The \ac{snr} is the same for training and evaluation.}
	\label{fig:pilots_quadriga_perSNR}
\end{figure}

	In Fig. \ref{fig:mse_quadriga_perSNR}, the performances of the above discussed estimators are evaluated for $N_p=18$ pilots over different \ac{snr} values. It can be seen that the ChannelNet has no performance improvement over linear interpolation, which is a consequence of the imperfect training data. The sample covariance and the naive \ac{gmm} approach based on linearly interpolated training data exhibit some performance improvement over linear interpolation in the low and medium \ac{snr} range, whereas for high \acp{snr}, an error floor because of the imperfections in the training data can be observed. The adapted \ac{gmm} version without structural constraints is performing close to the perfect \ac{csi} case in the low \ac{snr} range, but also seems to saturate for high \acp{snr}. In this \ac{snr}-region, the block-Toeplitz structured adapted \ac{gmm} performs especially well, showing performance gains over all other approaches even for high \ac{snr} values. This leads to the conclusion that the structural regularization, which is possible with the \ac{gmm}, is of great value when having noisy training data with missing entries.
	
	In Fig. \ref{fig:pilots_quadriga_perSNR}, we show a similar setup as above for a fixed \ac{snr} of $15$dB over varying numbers of pilots. Thereby, linear interpolation and the ChannelNet perform worst with a saturation at a high error floor, which was similarly observed before. It can be seen that the differences between the sample covariance, naive \ac{gmm}, and adapted (unconstrained) \ac{gmm} become more distinctive for higher numbers of pilots where the adapted version is outperforming them and has a decreasing gap to the perfect \ac{csi} based \ac{gmm}. This can be reasoned with a dominating systematic error for sparse pilot allocations. In contrast, the block-Toeplitz based adapted \ac{gmm} shows a substantial performance gain over the baselines even for low numbers of pilots, which is an effect of the structural regularization. As expected, the difference between the structurally unconstrained and constrained adapted \ac{gmm} decreases for higher numbers of pilots.

	\section{Conclusion} 
	In this letter, a \ac{gmm}-based robust channel estimator was proposed which can be applied in spatial and \ac{ofdm} systems. Thereby, the training data are purely comprised of sparsely allocated and noisy pilot observations, without perfect \ac{csi} \blue{samples}.
	%	We derived adapted training procedures of the underlying \ac{em} algorithm
	%	for both systems which mitigate the imperfections in the training data.
	%	 Finally, we imposed beneficial structural constraints to the covariances of the \ac{gmm}, which acts as a regularization. 
	Simulation results demonstrated that the performance of the proposed adapted \ac{gmm} estimator is close to the version which utilizes perfect training \ac{csi} with the same online complexity and memory overhead. 
	%	For both evaluated systems, the estimation performance is close to the case of perfect training \ac{csi}, even in the low \ac{snr} regime. 
	Additionally, state-of-the-art baselines are outperformed. Based on these findings, the superior robustness properties against imperfect training data of the generative model-aided estimator in contrast to regression-based \ac{ml} approaches were discussed.
	\blue{In future work, one may also account for imperfections due to interference which is not considered in this letter. }
	%	\vfill
	
	\bibliographystyle{IEEEtran}
	\bibliography{IEEEabrv,bibliography}

% Generated by IEEEtran.bst, version: 1.14 (2015/08/26)
\begin{thebibliography}{10}
\providecommand{\url}[1]{#1}
\csname url@samestyle\endcsname
\providecommand{\newblock}{\relax}
\providecommand{\bibinfo}[2]{#2}
\providecommand{\BIBentrySTDinterwordspacing}{\spaceskip=0pt\relax}
\providecommand{\BIBentryALTinterwordstretchfactor}{4}
\providecommand{\BIBentryALTinterwordspacing}{\spaceskip=\fontdimen2\font plus
\BIBentryALTinterwordstretchfactor\fontdimen3\font minus
  \fontdimen4\font\relax}
\providecommand{\BIBforeignlanguage}[2]{{%
\expandafter\ifx\csname l@#1\endcsname\relax
\typeout{** WARNING: IEEEtran.bst: No hyphenation pattern has been}%
\typeout{** loaded for the language `#1'. Using the pattern for}%
\typeout{** the default language instead.}%
\else
\language=\csname l@#1\endcsname
\fi
#2}}
\providecommand{\BIBdecl}{\relax}
\BIBdecl

\bibitem{8052521}
H.~Ye, G.~Y. Li, and B.-H. Juang, ``{Power of Deep Learning for Channel
  Estimation and Signal Detection in OFDM Systems},'' \emph{IEEE Wireless
  Commun. Lett.}, vol.~7, no.~1, pp. 114--117, 2018.

\bibitem{8640815}
M.~Soltani, V.~Pourahmadi, A.~Mirzaei, and H.~Sheikhzadeh, ``{Deep
  Learning-Based Channel Estimation},'' \emph{IEEE Commun. Lett.}, vol.~23,
  no.~4, pp. 652--655, 2019.

\bibitem{NeWiUt18}
D.~Neumann, T.~Wiese, and W.~Utschick, ``{Learning the {MMSE} Channel
  Estimator},'' \emph{{IEEE} Trans. Signal Process.}, vol.~66, no.~11, pp.
  2905--2917, Jun. 2018.

\bibitem{9842343}
M.~Koller, B.~Fesl, N.~Turan, and W.~Utschick, ``{An Asymptotically MSE-Optimal
  Estimator Based on Gaussian Mixture Models},'' \emph{IEEE Trans. Signal
  Process.}, vol.~70, pp. 4109--4123, 2022.

\bibitem{vae_chEst}
M.~Baur, B.~Fesl, M.~Koller, and W.~Utschick, ``{Variational Autoencoder
  Leveraged MMSE Channel Estimation},'' in \emph{56th Asilomar Conf. Signals,
  Syst., Comput.}, 2022.

\bibitem{QuaDRiGa1}
S.~{Jaeckel}, L.~{Raschkowski}, K.~{Börner}, and L.~{Thiele}, ``{QuaDRiGa: A
  3-D Multi-Cell Channel Model With Time Evolution for Enabling Virtual Field
  Trials},'' \emph{{IEEE} Trans. Antennas Propag.}, vol.~62, no.~6, pp.
  3242--3256, 2014.

\bibitem{Choi2005}
J.-W. Choi and Y.-H. Lee, ``{Optimum Pilot Pattern for Channel Estimation in
  {OFDM} Systems},'' \emph{IEEE Trans. Wireless Commun.}, vol.~4, no.~5, pp.
  2083--2088, Sep. 2005.

\bibitem{bookBi06}
C.~M. Bishop, \emph{Pattern Recognition and Machine Learning (Information
  Science and Statistics)}.\hskip 1em plus 0.5em minus 0.4em\relax Berlin,
  Heidelberg: Springer-Verlag, 2006.

\bibitem{gmm_structured}
B.~Fesl, M.~Joham, S.~Hu, M.~Koller, N.~Turan, and W.~Utschick, ``{Channel
  Estimation based on Gaussian Mixture Models with Structured Covariances},''
  in \emph{56th Asilomar Conf. Signals, Syst., Comput.}, 2022.

\bibitem{7051818}
D.~{Neumann}, M.~{Joham}, L.~{Weiland}, and W.~{Utschick}, ``{Low-Complexity
  Computation of LMMSE Channel Estimates in Massive MIMO},'' in \emph{Proc.
  19th Int. ITG Workshop Smart Antennas}, 2015.

\bibitem{bookRubin02}
R.~J.~A. Little and D.~B. Rubin, \emph{{Statistical Analysis with Missing
  Data}}.\hskip 1em plus 0.5em minus 0.4em\relax John Wiley \& Sons,
  Incorporated, 2002.

\end{thebibliography}

\end{document}